\documentclass[12pt]{amsart}
\usepackage{amsmath,amssymb,amsfonts,mathrsfs, amsthm}
\usepackage[margin=1 in,heightrounded=true,centering]{geometry}

\usepackage{amscd}
\usepackage{verbatim}
\usepackage{euscript}

\newcommand{\cal}{\mathcal}

\newtheorem{theorem}{Theorem}[section]

\newtheorem{lemma}[theorem]{Lemma}
\newtheorem{remark}[theorem]{Remark}

\begin{document}

\title[On static Poincar\'e-Einstein metrics]{On static Poincar\'e-Einstein metrics}

\author{Gregory J Galloway}
\address{Dept of Mathematics, University of Miami, Coral Gables, Coral Gables, FL 33146, USA}
\email{galloway@math.miami.edu}

\author{Eric Woolgar}
\address{Department of Mathematical and Statistical Sciences,
University of Alberta, Edmonton, Alberta, T6G 2G1, Canada}
\address{Theoretical Physics Institute, University of Alberta, Edmonton, Alberta, T6G 2G1, Canada}
\email{ewoolgar@ualberta.ca}

\date{}

\begin{abstract}
\noindent The classification of solutions of the static vacuum Einstein equations, on a given closed manifold or an asymptotically flat one, is a long-standing and much-studied problem. Solutions are characterized by a complete Riemannian $n$-manifold $(M,g)$ and a positive function $N$, called the lapse. We study this problem on Asymptotically Poincar\'e-Einstein $n$-manifolds, $n\ge 3$, when the conformal boundary-at-infinity is either a round sphere, a flat torus or smooth quotient thereof, or a compact hyperbolic manifold. Such manifolds have well-defined Wang mass, and are time-symmetric slices of static, vacuum, asymptotically anti-de Sitter spacetimes. By integrating a mildly generalized form of an identity used by Lindblom, Shen, Wang, and others, we give a mass formula for such manifolds. There are no solutions with positive mass. In consequence, we observe that either the lapse is trivial and $(M,g)$ is Poincar\'e-Einstein or the Wang mass is negative, as in the case of time symmetric slices of the AdS soliton. As an application, we use the mass formula to compute the renormalized volume of the warped product $(X,\gamma)\simeq (M^3,g)\times_{N^2} (S^1,dt^2)$.

We also give a mass formula for the case of a metric that is static in the region exterior to a horizon on which the lapse function is zero. Then the manifold $(X,\gamma)$ is said to have a ``bolt'' where the $S^1$ factor shrinks to zero length. The renormalized volume of $(X,\gamma)$ is expected on physical grounds to have the form of the free energy per unit temperature for a black hole in equilibrium with a radiation bath at fixed temperature. When $M$ is 3-dimensional and admits a horizon, we apply this mass formula to compute the renormalized volume of $(X,\gamma)$ and show that it indeed has the expected thermodynamically motivated form.

We also discuss several open questions concerning static vacuum asymptotically Poincar\'e-Einstein  manifolds.
\end{abstract}

\maketitle

\section{Introduction}
\setcounter{equation}{0}

\noindent Let $(M,g)$ be a complete Riemannian $n$-manifold and consider positive solutions $N:M\to {\mathbb R}$ of the \emph{static Einstein system}
\begin{eqnarray}
\label{eq1.1} N{\rm Ric}&=& \nabla^2 N +\frac{2\Lambda}{(n-1)}Ng\ ,\\
\label{eq1.2} \Delta N &=& -\frac{2\Lambda}{(n-1)}N\ ,
\end{eqnarray}
where $\nabla$ is the Levi-Civita connection of $g$, $\nabla^2$ is the Hessian, $\Delta:={\rm tr}_g \nabla^2$ is the scalar Laplacian, ${\rm Ric}$ is the Ricci tensor of $g$, and $\Lambda$ is a constant called the \emph{cosmological constant}. Note that a consequence of the above equations is that the scalar curvature is
\begin{equation}
\label{eq1.3}
R=2\Lambda\ .
\end{equation}
When $\Lambda<0$, we can rescale the metric to obtain $\Lambda=-\frac12 n(n-1)$. Then the system (\ref{eq1.1}, \ref{eq1.2}) is equivalent to the equation
\begin{equation}
\label{eq1.4}
\nabla^2 N -g\Delta N-N{\rm Ric}=0\ .
\end{equation}

Solutions $(M,g,N)$ of this system are sometimes called \emph{Killing initial data sets}. If a positive-$N$ solution to this system can be found, then the $(n+1)$-dimensional spacetime $({\mathbb R}\times M, -N^2dt^2\oplus g)$ is negative Einstein and \emph{globally static}; that is, it admits a nowhere vanishing, hypersurface orthogonal, timelike Killing vector field $\frac{\partial}{\partial t}$. Solutions with nonnegative $N$ are also of interest. When the zero set of $N$ is a closed hypersurface in $M$, then $({\mathbb R}\times M, -N^2dt^2\oplus g)$ is a static exterior black hole metric. The zero set is totally geodesic in $(M,g)$ with ``surface gravity'' $|dN|$ constant on the zero set  \cite[Lemma 2.1(i)]{GM}.\footnote
{The paper \cite{GM} assumes that $n=3$, but the proof of the quoted result is easily seen to be valid for $n\ge 3$.}

A theorem quoted in Lichn\'erowicz \cite{Lichnerowicz}
%
%
states that if $(M,g)$ is a complete, asymptotically flat 3-manifold obeying (\ref{eq1.1}, \ref{eq1.2}) with $\Lambda=0$ and $N\to 1$ at infinity, then $(M,g)$ is Euclidean 3-space. This has an obvious proof. It also follows from the positive mass theorem \cite{SY, Witten}. In the case of closed 3-manifolds, nontrivial solutions of (\ref{eq1.1}, \ref{eq1.2}) have been found \cite{Kobayashi, Lafontaine} and have been used as counter-examples to the Fischer-Marsden \cite{FM} conjecture, which held that no nontrivial solutions would exist except those about which the linearized scalar curvature operator is surjective. Shen \cite{Shen} found that Kobayashi's and Lafontaine's nontrivial solutions always contained a totally geodesic 2-sphere (that is, a 2-sphere \emph{horizon} in the sequel).

In this paper, we consider positive and nonnegative solutions of (\ref{eq1.4}) on manifolds which admit a notion of conformal infinity and have sectional curvatures approaching $-1$ sufficiently rapidly there. In particular, we take $(M,g)$ to be Asymptotically Poincar\'e-Einstein (APE). Such manifolds admit a conformal infinity which is the zero set of a local coordinate $x$ called a \emph{special defining function}, which obeys ${\tilde g}^{-1}(dx,dx)=1$ in a neighbourhood of conformal infinity, where ${\tilde g}:=x^2 g$. Thus, on such a neighbourhood, $x$ is a Gaussian normal coordinate for the conformally rescaled metric ${\tilde g}$. Furthermore, in this coordinate system, the Einstein equations are enforced order-by-order on the coefficients ${\tilde a}_{[p]}$ in the expansion ${\tilde g}=\sum_p {\tilde a}_{[p]}x^p$ up to (but not including) order $x^n$. This condition fully determines ${\tilde a}_{[0]}, \dots, {\tilde a}_{[n-2]}$ and ${\rm tr}_{{\tilde a}_{[0]}}{\tilde a}_{[n-1]}$ in terms of the \emph{Dirichlet data} ${\tilde g}(x=0)\equiv {\tilde a}_{[0]}$ and, in fact, up to this order the odd coefficients ${\tilde a}_{[2p+1]}$ vanish. This is known as an even Fefferman-Graham expansion; see \cite{BMW} for more detail.

When conformal infinity carries either the round sphere metric, a compact flat metric, or a compact hyperbolic metric, APEs have well-defined Wang mass \cite{Wang1}. If an APE is exactly Poincar\'e-Einstein, and if its conformal infinity is one of the above types which admit a Wang mass, then that mass is zero \cite{AD}. In this note, we generalize that result as follows:

\begin{theorem}\label{theorem1.1}
Let $(M,g)$ be a complete APE with conformal infinity either a round $(n-1)$-sphere, flat $(n-1)$-torus or a smooth quotient thereof, or compact hyperbolic $(n-1)$-manifold. Let $N>0$ solve equation (\ref{eq1.4}) with $|dN|\to 1$ on approach to infinity.
\begin{enumerate}
\item[{\bf (a)}] If $\partial M$ is empty, then the Wang mass of $g$ is given by
\begin{equation}
\label{eq1.5}m=-\frac{1}{8\pi(n-2)}\int_M N\left \vert Z \right \vert^2 dV(g)\ ,
\end{equation}
where\footnote
{Since by (\ref{eq1.3}) the Ricci scalar is constant, then $Z$ as given by (\ref{eq1.6}) equals the tracefree Ricci tensor.}
\begin{equation}
\label{eq1.6} Z:={\rm Ric} +(n-1)g\ .
\end{equation}
\item[{\bf (b)}] If $M$ has a non-empty boundary $\partial M=:{\cal H}=\sqcup_i {\cal H}_{i}$ comprised of finitely many disjoint compact connected components ${\cal H}_{i}$ such that $N\equiv 0$ on ${\cal H}$ then
\begin{equation}
\label{eq1.7}m=-\frac{1}{8\pi(n-2)}\int_M N\left \vert Z \right \vert^2 dV(g)+\sum_i\frac{\vartheta_i}{16\pi} \left [ (n-1) \vert {\cal H}_{i}\vert +\frac{1}{(n-2)}\int_{{\cal H}_{i}}S dV_{{\cal H}_{i}}\right ]\ ,
\end{equation}
where $\vert {\cal H}_{i}\vert$ is the surface area of ${\cal H}_{i}$, $S$ is the intrinsic scalar curvature of ${\cal H}_{i}$, and $\vartheta_i=\vert dN \vert_{{\cal H}_{i}}$ is a constant on ${\cal H}_{i}$, known as the \emph{surface gravity}.
\begin{enumerate}
\item[{\bf (i)}] In the special case of a \emph{cold horizon}, defined by $\vartheta=0$, the mass is nonpositive and we recover (\ref{eq1.5}).
\item[{\bf (ii)}] In the $n=3$ case, we have
\begin{equation}
\label{eq1.8}m=-\frac{1}{8\pi}\int_M N\left \vert Z \right \vert^2 dV(g)+\sum_i\frac{\vartheta_i}{8\pi}\left ( \vert {\cal H}_{i}\vert +2\pi \chi({\cal H}_{i})\right )\ ,
\end{equation}
where $\chi({\cal H}_{i})=2(1-g_i)$ is the Euler characteristic of ${\cal H}_{i}$, and $g_i$ is the genus of ${\cal H}_{i}$.
\end{enumerate}
\end{enumerate}
\end{theorem}

This also generalizes a result of Chru\'sciel and Simon \cite{CS}, who observed that $m<0$ for complete solutions of (\ref{eq1.1}, \ref{eq1.2}) with APE asymptotics in the particular case of $n=3$ dimensions and compact hyperbolic conformal infinity. We obtain our result essentially by following a computation of Wang \cite{Wang2}, who was concerned with the $k=1$ case. He was able to show that amongst $k=1$ APEs with spinor structure, the argument that leads us to Theorem \ref{theorem1.1}, when combined with the positive mass theorem, implies that anti-de Sitter spacetime is the unique complete spin manifold solving of (\ref{eq1.1}, \ref{eq1.2}) with $N>0$ globally.

Examples of solutions of the system (\ref{eq1.1}, \ref{eq1.2}) with cold horizons are provided by the extreme ``topological'' black holes described in (\cite{Mann}, \cite{BLP}). We are concerned with vacuum metrics only, and therefore while charged extreme black holes also admit cold horizons, only uncharged cold horizons provide examples for our theorem.\footnote
{In particular, an example is obtained by setting $k=b=-1$, $q=0$, and $m=-3^{-3/2}\ell$ in equation (7) of \cite{Mann}.}
We also note that examples of negative mass complete solutions of (\ref{eq1.1}, \ref{eq1.2}) with empty $\partial M$ are known. They are time-symmetric slices of so-called AdS solitons, and are discussed briefly in section 5.

In the horizon-free case, we have $N>0$, so we can consider the Riemannian warped product Poincar\'e-Einstein metric $\gamma:=N^2dt^2\oplus g$ on $X\simeq S^1\times M$. In this case, our mass formula (\ref{eq1.7}) yields a novel application. By a simple calculation outlined in Section 4, Theorem \ref{theorem1.1} implies that the $L^2$ norm of the Riemann tensor of $\gamma$, renormalized by subtraction of a dimension-dependent constant, equals the mass of $(M,g)$. In particular, if $n=3$, then this observation can be used to determine the \emph{renormalized volume} \cite{HS, Graham} ${\rm RenV}(X,\gamma)$, via a formula of Anderson \cite{Anderson}, in terms of the Wang mass $m$ of $(M,g)$.

This application of the mass formula becomes more interesting when a horizon is present.  If, in four spacetime dimensions, we assume the horizon is connected then, under physically natural circumstances \cite{EGP, BG},\footnote
{While \cite{EGP} discusses only the asymptotically flat case, its analysis is valid in the present case as well.}
$M$ must have topology $[a,\infty)\times {\cal H}$ for some $a>0$ and some surface $\Sigma$, and hence $X$ will have topology $X\simeq {\mathbb R}^2\times {\cal H}$. The horizon then contributes a boundary term, leading to a strikingly simple formula for ${\rm RenV}(X,\gamma)$ with an obvious interpretation in black hole thermodynamics.

\begin{theorem}\label{theorem1.2} Let $(X,\gamma)$ be a Poincar\'e-Einstein 4-manifold, with a hypersurface-orthogonal Killing vector $K=\frac{\partial}{\partial t}$.
\begin{enumerate}
\item[{\bf (a)}] If $(X,\gamma)=(S^1\times M^3, N^2dt^2\oplus g)$ such that $(M^3,g)$ is complete, is APE with conformal infinity having constant sectional curvature $k\in \{ -1,0,1\}$, and has Wang mass $m$, then the renormalized volume of $(X,\gamma)$ is given by
\begin{equation}
\label{eq1.9}
{\rm RenV}(X,\gamma)= \frac{8\pi}{3}m\beta \le 0 \ .
\end{equation}
where $\beta$ is the circumference of a Killing orbit at infinity as measured in the conformal metric.
\item[{\bf (b)}] If $M^3\simeq [a,\infty)\times {\cal H}$ for some $a>0$ and connected manifold ${\cal H}$, and if $\{a\}\times {\cal H}$ is the zero set of $N$, then $X\simeq {\mathbb R}^2\times {\cal H}$ and the renormalized volume of $(X,\gamma)$ is given by
\begin{equation}
\label{eq1.10}
{\rm RenV}(X,\gamma)= \frac{8\pi}{3} \left [ m\beta-\frac14 \vert{\cal H}\vert \right ]  \ .
\end{equation}
\end{enumerate}
\end{theorem}

Equation (\ref{eq1.10}) is in fact a familiar thermodynamic formula. Consider a static black hole of mass $m$ with horizon ${\cal H}$ in equilibrium with radiation at temperature $T=1/\beta$ (the canonical ensemble). Then the expectation value of the energy of the system is $\langle E \rangle=m$ and the entropy of the system is famously given by $S=\frac14 |{\cal H}|$. One expects then to have the formula
\begin{equation}
\label{eq1.11} I = \frac{1}{T} \langle E \rangle -S = \beta m -\frac14 |{\cal H}|\ ,
\end{equation}
where $I$ is the gravitational action $-\frac{1}{32\pi}\int_X R_{\gamma} dV(\gamma)$ of the black hole metric, Wick rotated to Riemannian signature. From equation (\ref{eq1.11}), $I$ plays the role of the \emph{free energy per unit temperature}. By comparing this quantity, evaluated on different static metrics with the same value of $\beta$, one can construct a free energy diagram and quantify the energy liberated in phase transitions between these metrics, as well as any energy barriers to be overcome as a phase transition proceeds. However, since $R=-12$ for a Poincar\'e-Einstein 4-metric, we have $I=\frac{3}{8\pi}{\rm vol}(X)$, which is divergent. The renormalized volume was introduced as a method of rendering the action finite and well-defined \cite{HS}. It is therefore to be expected on physical grounds, though from a purely geometric perspective it appears startling, that equation (\ref{eq1.10}) yields
\begin{equation}
\label{eq1.12}
\frac{3}{8\pi}{\rm RenV}(X,\gamma) = \frac{1}{T}\langle E \rangle -S\ .
\end{equation}
In summary:

\begin{remark}[Thermodynamic interpretation of renormalized volume]\label{remark1.3}
When $(X,\gamma)$ is as described in Theorem \ref{theorem1.2}.(b), the renormalized volume ${\rm RenV}(X,\gamma)$ equals the free energy per unit temperature of the static black hole got by Wick rotating $(X,\gamma)$, in equilibrium with radiation at temperature $T=1/\beta$.
\end{remark}

It is not at all clear whether this interpretation can be extended to static black holes in the presence of matter, such as the Reissner-Nordstrom-AdS family. This issue is under investigation.

This paper is organized as follows. In Section 2, we recall Asymptotically Poincar\'e-Einstein boundary conditions and the Wang mass. In Section 3, we derive a simple identity of divergence form and integrate it over the manifold to prove Theorem \ref{theorem1.1}. Much of this section follows the argument given first by Wang \cite{Wang2} in a less general context, which was key to his uniqueness proof for anti-de Sitter spacetime. In section 3.3, we depart from this and use a different method based on the maximum principle to prove that, in the setting of Theorem \ref{theorem1.1}.(a), the mass \emph{aspect} is pointwise nonpositive. In Section 4, we prove Theorem \ref{theorem1.2}. We give a nontrivial example of Part (a) of that theorem in Section 4.3. In Section 5, we discuss several open problems for static APE manifolds, some of which are highly nontrivial.

\medskip
\noindent\emph{Acknowledgements.} The work of GJG was  supported by NSF grant DMS--1313724 and by a grant from the Simons Foundation (Grant No 63943). The work of EW was supported by NSERC Discovery Grant RGPIN 203614. Both authors wish to express their gratitude to the University of Science and Technology of China for hosting the 2013 Conference on Geometric Analysis and Relativity, at which this work was conceived, and to Piotr Chru\'sciel for a discussion of \cite{CS} at that time. EW thanks Don Page for discussions on black hole thermodynamics.

\section{APEs and Wang's mass}
\setcounter{equation}{0}

\noindent The metrics we consider must meet three criteria. First, they must be \emph{conformally compactifiable}, meaning that they admit a notion of conformal infinity defined as the locus $x=0$, to which the conformal metric ${\tilde g}:=x^2 g$ extends. Second, we require that $|dx|_{\tilde g}=1$ at conformal infinity. The $C^2$ smoothness of the conformal metric (which we will take to be $C^{\infty}$) allows this, and it follows that the sectional curvatures of $g$ must asymptote to $-1$, so such metrics are called \emph{asymptotically hyperbolic}. We can then extend the condition $|dx|_{\tilde g}=1$ to a neighbourhood of conformal infinity since $|dx|_{\tilde g}=1$ is a non-characteristic first-order differential equation, whose local solution $x$ therefore exists. This yields a Gaussian normal coordinate system for that neighbourhood. Then $x$ is called a \emph{special defining function} for conformal infinity. And third, the metric must have a well-defined mass.

The mass of asymptotically hyperbolic manifolds was first defined by Wang \cite{Wang1} in the special case where conformal infinity was a round sphere, but it easily generalizes to the three cases listed in the Introduction. We will index these cases by $k$, the sectional curvature of the conformal boundary-at-infinity $\partial_{\infty}M$, so that $k=1$ represents the case where $\partial_{\infty}M$ is the round metric $g^{(+1)}:=g(S^{n-1},{\rm can})$, $k=0$ represents the case where $\partial_{\infty}M$ carries a flat torus metric $g^{(0)}:=\delta$, and $k=-1$ denotes the case where $\partial_{\infty}M$ is a compact hyperbolic manifold with metric $g^{(-1)}$. Specifically, we now require that
\begin{equation}
\begin{split}
\label{eq2.1} g=&\, \frac{1}{f_{(k)}^2(r)} \left ( dr^2+g^{(k)}+\frac{1}{n}\kappa r^n+{\cal O}(r^{n+1}) \right )\ ,\\
f_{(k)}=&\, \begin{cases} \sinh r,& k=+1,\\ r,& k=0,\\ \sin r, & k=-1\ .\end{cases}
\end{split}
\end{equation}
For such metrics, the Wang mass is defined to be
\begin{equation}
\label{eq2.2} m:=\frac{1}{16\pi}\int_{\partial_{\infty}M} {\rm tr}_{g^{(k)}}\kappa\ dV(g^{(k)})\ ,
\end{equation}
where $\kappa$ is a symmetric $(0,2)$-tensor on $\partial_{\infty}M$.

The pre-factor $\frac{1}{16\pi}$ does not appear in \cite{Wang1}. We include it so as to agree with the mass used in $(3+1)$-dimensional asymptotically anti-de Sitter general relativity. If Newton's constant $G$ is not set to $1$, the normalization would then be $\frac{1}{16\pi G}$. An alternative normalization would be to divide $m$ by $4{\rm vol}_{g^{(k)}}$, which is of course $16\pi$ when $n=3$ and $k=1$. This, however, would have a disadvantage in the $k=0$ case where there are non-isometric Horowitz-Myers geons (time-symmetric slices of AdS solitons \cite{HM}) whose normalized masses would then be the same (cf \cite[Section 1]{BW}), so we will not do this.

We note here that conformal infinity is the locus $r=0$ but $r$ is not a special defining function since $|dr|_{r^2g}\neq 1$ on any open domain $r<\epsilon$. To obtain a special defining function, we solve
\begin{equation}
\label{eq2.3}\frac{dx}{x}=\frac{dr}{f_{(k)}}\ ,
\end{equation}
subject to the condition that $x= 0$ when $r=0$. Then the metric (\ref{eq2.1}) can be written as
\begin{equation}
\label{eq2.4}
g=\frac{1}{x^2} \left [ dx^2 + \left ( 1-kx^2/4\right )^2 g^{(k)} +\frac{1}{n}\kappa x^n +{\cal O}(x^{n+1})\right ]\ .
\end{equation}
This form is precisely what one obtains by following the Fefferman-Graham \cite{FG} method of applying the Einstein equations (for $g\equiv {\tilde g}/x^2$) order-by-order, up to order $x^{n-1}$ inclusive, to the formal expansion  ${\tilde g}= \sum a_{[n]}x^n$, subject to the Dirichlet condition ${\tilde g}(0) \equiv a_{[0]} =g^{(k)}$. Therefore, the metrics we consider are precisely the \emph{Asymptotically Poincar\'e-Einstein} metrics (APEs, see \cite{BMW}) with one of the constant curvature conformal infinities.

For use in the sequel, we note that the shape operator of hypersurfaces of constant $x$ is easy to compute from (\ref{eq2.4}). Computed with respect to the \emph{inward-pointing} normal vector field $\nu:=x\frac{\partial}{\partial x}$, it has components
\begin{equation}
\label{eq2.5}
A^{\alpha}{}_{\beta}= - \left ( \frac{1+\frac{kx^2}{4}}{1-\frac{kx^2}{4}}\right ) \delta^{\alpha}_{\beta}
+\frac12 \kappa^{\alpha}{}_{\beta}x^n +{\cal O}(x^n)\ ,
\end{equation}
where the Greek indices run over the tangent space to $\partial_{\infty}M$ so that $\alpha,\beta\in \{ 2,\dots,n\}$,  $\delta^{\alpha}_{\beta}$ denotes the components of the $(n-1)\times (n-1)$ identity matrix, and $\kappa^{\alpha}{}_{\beta}:=g^{(k)\alpha\gamma}\kappa_{\gamma\beta}$. The mean curvature of these hypersurfaces is then
\begin{equation}
\label{eq2.6}
H={\rm tr} A = -(n-1)\left ( \frac{1+\frac{kx^2}{4}}{1-\frac{kx^2}{4}}\right )+\frac12 \left ( {\rm tr}_{g^{(k)}} \kappa \right ) x^n +{\cal O}(x^{n+1})\ .
\end{equation}

\section{Proof of Theorem \ref{theorem1.1}}
\setcounter{equation}{0}

\subsection{Divergence identity}

\begin{lemma}\label{lemma3.1} If $N$ is a solution of (\ref{eq1.1}, \ref{eq1.2}) then
\begin{equation}
\label{eq3.1} {\rm Div} \left [ \frac{1}{N}\nabla\left ( |dN|^2-N^2+k \right ) \right ] = 2N |Z|^2\ .
\end{equation}
\end{lemma}

A 3-dimensional form of this identity appeared in \cite{Lindblom} and in several works since. A related but much more complicated identity was found by Robinson as early as 1975 and used to prove a uniqueness theorem for the Kerr metric \cite{Robinson}. The $n$-dimensional version occurs in \cite{Shen} and was used by \cite{Wang2} to prove his uniqueness result.

\begin{proof} We proceed by direct calculation and application of equations (\ref{eq1.1}, \ref{eq1.2}) and the contracted second Bianchi identity, which in the present case yields $\nabla^i R_{ij}=\frac12\nabla_jR=0$ since $R=-n(n-1)$, and in particular $\nabla^i \left ( R_{ij}+(n-1)g_{ij} \right ) =0$.
\begin{equation}
\label{eq3.2}
\begin{split}
{\rm Div} \left [ \frac{1}{N}\nabla \left ( |dN|^2-N^2+k \right ) \right ] =&\, 2\nabla^i \left [ \left
( \frac{\nabla_i \nabla_j N}{N}-g_{ij} \right ) \nabla^j N\right ]\\
=&\, 2\nabla^i \left [ \left ( R_{ij}+(n-1)g_{ij} \right ) \nabla^j N \right ]\\
=&\, 2 \left ( R_{ij}+(n-1)g_{ij} \right ) \nabla^i \nabla^j N \\
=&\, 2N\left ( R_{ij}+(n-1)g_{ij} \right ) \left ( R^{ij}+ng^{ij}\right )\\
=&\, 2N\left ( R_{ij}+(n-1)g_{ij} \right ) \left ( R^{ij}+(n-1)g^{ij}\right )\\
&\, +2N\left ( R_{ij}+(n-1)g_{ij} \right )g^{ij} \\
=&\, 2N\left \vert {\rm Ric}+(n-1)g \right \vert^2 N+2N \left ( R+ n(n-1) \right )\\
=&\, 2N \left \vert Z \right \vert^2 \ ,
\end{split}
\end{equation}
where in the last equality we used that $R=2\Lambda=-n(n+1)$.\end{proof}

We remark that the quantity inside the operator on the left-hand side of (\ref{eq3.1}) has a simple interpretation. Let ${\tilde g}=g/N^2$. This is sometimes called the \emph{Fermat metric}. Applying (\ref{eq1.1}, \ref{eq1.2}) to the standard formula for the behaviour of scalar curvature under a conformal transformation, one can check that the scalar curvature of ${\tilde g}$ is given by
\begin{equation}
\label{eq3.3}
{\tilde R}=-n(n-1)\left ( |dN|^2-N^2\right )\ ,
\end{equation}
so Lemma \ref{lemma3.1} can be written as
\begin{equation}
\label{eq3.4}
{\rm Div} \left [ \frac{1}{N}\nabla \left ({\tilde R}-n(n-1)k\right ) \right ] = -2n(n-1)N|Z(g)|^2\ .
\end{equation}

\subsection{The proof of Theorem \ref{theorem1.1}}

\begin{proof}[Proof of Part (a).] Now consider the manifold $M_{\epsilon}:= M\backslash \{ x\le\epsilon \}$, the submanifold of $M$ consisting of all points except those ``$\epsilon$-close to conformal infinity''. The boundary $\partial M_{\epsilon}=: \partial_{1/\epsilon}M$ of this set is the hypersurface $x=\epsilon$; the notation indicates that as $\epsilon\to 0$ then $\partial_{1/\epsilon}M$ is replaced by the boundary-at-infinity $\partial_{\infty}M$. If we integrate the identity (\ref{eq3.1}) over $M_{\epsilon}$ and use the divergence theorem, we obtain
\begin{equation}
\label{eq3.5}
\int_{\partial_{1/\epsilon}M} \frac{1}{N} \nabla_{\nu} \left ( |dN|^2-N^2+k \right ) dV(h) = 2\int_{M_{\epsilon}} N \left \vert Z \right \vert^2 dV(g)\ ,
\end{equation}
where $\nu$ is the \emph{outward pointing} unit normal field (pointing toward infinity) and $dV(h)$ is the volume element of the metric
\begin{equation}
\label{eq3.6}
h:= \frac{1}{\epsilon^2} \left ( 1-k\epsilon^2/4\right )^2 g^{(k)}+\frac{1}{n}\kappa \epsilon^{n-2}+{\cal O}(\epsilon^{n-1})
\end{equation}
induced on $\partial_{1/\epsilon}M$ by $g$.

On the other hand, we compute
\begin{equation}
\label{eq3.7}
\begin{split}
\frac{1}{N} \nabla_{\nu} \left ( |dN|^2-N^2+k \right ) =&\, \frac{2}{N}\left [ \left ( \nabla^k N \right ) \left ( \nabla_{\nu}\nabla_k N\right ) -N\nabla_{\nu} N\right ]\\
=&\, 2\left ( \nabla^k N\right ) \left ( R_{jk} + ng_{jk} \right ) \nu^j -2\nabla_{\nu} N\\
=&\, 2\left ( \nabla^k N\right ) \left ( R_{jk} + (n-1)g_{jk} \right ) \nu^j\\
=&\, 2Z(\nu,\nabla N)\\
=&\, 2Z(\nu,\nu) \vert \nabla N\vert \left ( 1+{\cal O}(\epsilon)\right )\ .
\end{split}
\end{equation}
Thus we obtain
\begin{equation}
\label{eq3.8}
\int_{\partial_{1/\epsilon}M} Z(\nu,\nu) \vert \nabla N\vert \left ( 1+{\cal O}(\epsilon)\right ) dV(h) = \int_{M_{\epsilon}} N \left \vert Z \right \vert^2 dV(g)\ .
\end{equation}

Using the Gauss-Codazzi equation, on the level set $\partial_{1/\epsilon}M$ of $x$ we have
\begin{equation}
\label{eq3.9}
Z(\nu,\nu)={\rm Ric}(\nu,\nu)+(n-1)=\frac12 \left ( R-S+H^2-|A|^2 \right ) +n-1\ ,
\end{equation}
where $R$ and $S$ are the intrinsic scalar curvatures of $M$ and $\partial_{1/\epsilon}M$ respectively. Using (\ref{eq2.5}) and (\ref{eq2.6}) on the $x=\epsilon$ hypersurface, then
\begin{equation}
\label{eq3.10}
\begin{split}
Z(\nu,\nu)=&\, \frac12 (n-1)(n-2)\left [ \left ( \frac{1+\frac{k\epsilon^2}{4}}{1-\frac{k\epsilon^2}{4}} \right )^2 -1 \right ] -\frac12 S -\frac12 (n-2)\left ( {\rm tr}_{g^{(k)}}\kappa\right ) \epsilon^n +{\cal O}(\epsilon^{n+1})\\
=&\, \frac12 (n-1)(n-2)\frac{k\epsilon^2}{\left ( 1-\frac{k\epsilon^2}{4}\right )^2}  -\frac12 S -\frac12 (n-2)\left ( {\rm tr}_{g^{(k)}}\kappa\right ) \epsilon^n +{\cal O}(\epsilon^{n+1})\ .
\end{split}
\end{equation}
We must evaluate $S$, the scalar curvature of the metric $h$. To necessary order, it suffices to write that
\begin{equation}
\label{eq3.11}
h=\frac{1}{\epsilon^2} \left ( 1-k\epsilon^2/4\right )^2 \left ( g^{(k)}+{\cal O}(\epsilon^n)\right )
\end{equation}
and
\begin{equation}
\label{eq3.12}
S[h]=
\frac{\epsilon^2}{\left ( 1-k\epsilon^2/4\right )^2}\left ( S[g^{(k)}]+{\cal O}(\epsilon^n) \right )
=\frac{k(n-1)(n-2)\epsilon^2}{\left ( 1-k\epsilon^2/4\right )^2} +{\cal O}(\epsilon^{n+2})\ .
\end{equation}
Then (\ref{eq3.10}) yields
\begin{equation}
\label{eq3.13}
Z(\nu,\nu)=-\frac12 (n-2)\left ( {\rm tr}_{g^{(k)}}\kappa\right )\epsilon^n +{\cal O}(\epsilon^{n+1})\ .
\end{equation}
We insert this into the left-hand side of (\ref{eq3.8}) to obtain
\begin{equation}
\label{eq3.14}
-\frac12 (n-2)\int_{\partial_{1/\epsilon}M} \left [ \left ( {\rm tr}_{g^{(k)}}\kappa\right ) \epsilon^n +{\cal O}(\epsilon^{n+1})\right ] \vert \nabla N\vert \left ( 1+{\cal O}(\epsilon)\right ) dV(h) = \int_{M_{\epsilon}} N \left \vert Z \right \vert^2 dV(g)\ .
\end{equation}
Finally, to prove the theorem, take $\epsilon\to 0$, noting that then $\nabla N\to \nu$ and so $\vert \nabla N\vert\to 1$, and $dV(h)=\epsilon^n dV(g^{(k)})+{\cal O}(\epsilon^{n-1})$. This yields (\ref{eq1.5}). \end{proof}

A much quicker proof in the $n=3$ case is inspired by the observation that our equation (\ref{eq3.3}) is equation (III.15) of \cite{CS} when $n=3$. Simply use equation (\ref{eq3.3}) to replace the left-hand side of (\ref{eq3.5}) by an integral over $\partial_{1/\epsilon} M$ of $\nabla_{\nu} {\tilde R}$ and then use the Chru\'sciel-Simon mass formula \cite[equation (V.23)]{CS}. Our more detailed derivation, however, clearly illustrates the role of the APE assumption and resulting expansion for $g$ and, under the APE assumption, holds manifestly in all dimensions.

\begin{proof}[Proof of Part (b).]
This time we must account for ``finitely distant'' boundary components ${\cal H}_{i}$, defined as the locus $N=0$. To avoid division by zero, we displace those components slightly into $M$, say by moving each component ${\cal H}_{i}$ a distance $\epsilon$ along the geodesic congruence orthogonal to it; we call the displaced hypersurface ${\cal H}_{i,\epsilon}$ and define ${\cal H}_{\epsilon}:=\cup_i {\cal H}_{i,\epsilon}$. We then redefine $M_{\epsilon}$ to be the connected submanifold of $M$ whose boundary is $\partial M_{\epsilon}:={\cal H}_{\epsilon}\cup \partial_{1/\epsilon}M$ where, as before, $\partial_{1/\epsilon}M=\{ p\in M; x(p)=\epsilon\}$ and $x$ is as usual our special defining function for the boundary-at-infinity.

Then (\ref{eq3.5}) is replaced by
\begin{equation}
\label{eq3.15}
\begin{split}
\int_{\partial_{1/\epsilon}M} \frac{1}{N} \nabla_{\nu} \left ( |dN|^2-N^2+k \right ) dV(h)- \sum_i \int_{{\cal H}_{i,\epsilon}} \frac{1}{N} \nabla_{\nu} \left ( |dN|^2-N^2+k \right ) dV(h)\\
= 2\int_{M_{\epsilon}} N \left \vert Z \right \vert^2 dV(g)\ .
\end{split}
\end{equation}
On ${\cal H}_{\epsilon}$ the unit normal field is chosen to point into $M_{\epsilon}$, so it again points toward infinity.

Equation (\ref{eq3.7}) remains valid at the inner boundary ${\cal H}_{\epsilon}$. The Gauss-Codazzi relation (\ref{eq3.9}) also holds. Furthermore, by \cite[Theorem 2.1.i]{GM}, ${\cal H}$ is necessarily an embedded, totally geodesic hypersurface and $|\nabla N|=:\vartheta_i$ is constant on each component ${\cal H}_{i}$, so we can write $|\nabla N|=\vartheta_i\left ( 1+{\cal O}(\epsilon) \right )$ on ${\cal H}_{i,\epsilon}$ and then $\vartheta_i$ can come outside the integral. Taking $\epsilon\to 0$, the Gauss-Codazzi relation becomes simply $2Z(\nu,\nu)=-(n-1)(n-2)-S$. Then the first term  in (\ref{eq3.15}) yields $-16\pi (n-1) m$ while the second term reduces to $\sum_i \left [ (n-1)(n-2)|{\cal H}_{i}|+\int_{{\cal H}_{i}} S\, dV_{{\cal H}_{i}}\right ] \vartheta_i$, and so we obtain (\ref{eq1.7}).\end{proof}

\subsection{The mass aspect}

\noindent We consider now the spacetime metric $-N^2 dt^2+g$ constructed from $N$ and $g$. We wish to apply \emph{asymptotically anti-de Sitter} boundary conditions to this metric in order to understand the properties of the \emph{mass aspect function}
\begin{equation}
\label{eq3.16} \mu:={\rm tr}_{g^{(k)}}\kappa\ ,
\end{equation}
defined on the boundary-at-infinity. Since the metric is static, one way to proceed is to impose the APE condition on the Riemannian metric ${\bar g}=N^2dt^2+g$ on a neighbourhood of infinity in $X^{n+1}$. Indeed, this metric will be not merely APE but exactly Poincar\'e-Einstein, but in what follows we will not need to apply the Einstein equations beyond APE order.

Because the spacetime metric is static, if $x$ is a special defining function for $(M^n,g)$ then it is also a special defining function for $(X^{n+1},{\bar g})$. We may take $X^{n+1}\simeq I\times M^n$ where $I$ is $S^1$ or ${\mathbb R}$. Since $g$ is still subject to the APE condition on $M$, its leading terms up to order $x^{n-3}$ inclusive (order $x^{n-1}$ inclusive in the standard counting which refers to the conformal metric $x^2g$) are determined as before. At this stage, the metric on $X$ is
\begin{equation}
\label{eq3.17}
\begin{split}
{\bar g} =&\, \frac{1}{x^2} \left [ V dt^2+dx^2+\left ( 1-\frac{kx^2}{4} \right )^2 g^{(k)} +\frac{1}{n}x^n\kappa + {\cal O}(x^{n+1})\right ]\ ,\\
V:=&\, x^2 N^2\ .
\end{split}
\end{equation}
Now the APE condition for $(X,{\bar g})$ can be applied to determine $V$, and thus $N$, order-by-order. First, since $\dim X=n+1$, this condition fixes all the coefficients $a_{[j]}$ in the expansion ${\bar g}=\frac{1}{x^2}\sum_{j=0}^{\infty} a_{[j]}x^j$ up to $j=n-1$ inclusive once $a_{[0]}$ is specified. This is in fact simple, and yields
\begin{equation}
\label{eq3.18}
V= \left ( 1+\frac{kx^2}{4} \right )^2+v_{[n]}x^n+{\cal O}(x^{n+1})\ ,
\end{equation}
and we note that when $n=3$ the above expression contains an explicit $k^2x^4/16$ term which belongs to ${\cal O}(x^{n+1})$ and so can be ignored. The coefficient $v_{[n]}$ is as yet undetermined, but we can find it using the Einstein equation at order $j=n$. Unlike at lower orders, at this order the Einstein equation fixes only the trace of the $j=n$ term, but that suffices. It reads
\begin{equation}
\label{eq3.19}
v_{[n]}+\frac{1}{n}{\rm tr}_{g^{(k)}}\kappa=0\ .
\end{equation}
Then, since $N^2=V/x^2$, we have
\begin{equation}
\label{eq3.20}
\begin{split}
&\, N^2=\frac{1}{x^2}\left [ \left ( 1+\frac{kx^2}{4}\right )^2 -\frac{x^n}{n}{\rm tr}_{g^{(k)}}\kappa +{\cal O}(x^{n+1})\right ]\ , \\
\Rightarrow &\, dN= -\frac{1}{x^2}\left [ 1-\frac{kx^2}{4} +\frac{(n-1)}{2n}x^n{\rm tr}_{g^{(k)}}\kappa +{\cal O}(x^{n+1})\right ] dx \ ,\\
\Rightarrow &\, |dN|^2-N^2+k= x^{n-2} {\rm tr}_{g^{(k)}}\kappa +{\cal O}(x^{n-1})=\mu x^{n-2}+{\cal O}(x^{n-1})\ .
\end{split}
\end{equation}
Recalling (\ref{eq3.3}), one can now obtain the Chru\'sciel-Simon mass formula. Furthermore, we can now prove the following result.

\begin{theorem}\label{theorem3.2} Let $(M,g,N)$ be a solution of (\ref{eq1.1}, \ref{eq1.2}) such that $-N^2 dt^2\oplus g$ is static Einstein and asymptotically anti-de Sitter in the above sense, with $(M,g)$ complete and $N>0$. Let $\partial_{\infty}M$ be the boundary-at-infinity of $(M,g)$. Then $\mu \le 0$ pointwise on $\partial_{\infty}M$.
\end{theorem}

\begin{proof}
Let
\begin{equation}
\label{eq3.21}
f:=|dN|^2-N^2+k= \mu x^{n-2}+{\cal O}(x^{n-1})\ .
\end{equation}
Then $f\to 0$ as $x\searrow 0$. If there were a point of $\partial_{\infty}M$ where $\mu>0$, then by continuity there would be a ``nearby'' point $p\in M$ where $f(p)>0$. Then $f$ would achieve a positive maximum in $M$. But from (\ref{eq3.2}) we may write
\begin{equation}
\label{eq3.22}
\left ( \Delta -\frac{1}{N}\nabla_{\nabla N} \right ) f =2N^2|Z|^2\ge 0\ ,
\end{equation}
and then by the Hopf strong maximum principle \cite[Theorem 3.5, p 35]{GT}, $f$ would necessarily be identically zero in $M$, contradicting $\mu>0$. Hence $\mu\le 0$.
\end{proof}

\begin{remark}\label{remark3.3} The argument above also implies that $f\le 0$ on $M$ and so, by (\ref{eq3.3}), the scalar curvature of the Fermat metric obeys ${\tilde R}\ge n(n-1)k$.
\end{remark}

\begin{remark}\label{remark3.4}
The conclusion of Theorem \ref{theorem3.2} remains valid in the horizon case (where the horizon is given by the zero set of $N$), provided $k = -1$ and the surface gravity is sufficiently small; specifically, $|dN| \le 1$ so that from
(\ref{eq3.21}) $f \le 0$ at the horizon. Then if $f$ had a positive maximum, it would necessarily occur at an interior point. As before, this contradicts the strong maximum principle.
\end{remark}

\section{Renormalized volume and Theorem \ref{theorem1.2}}
\setcounter{equation}{0}

\subsection{The Pfaffian of the curvature 2-form}

\noindent Now consider the $(n+1)$-dimensional Riemannian manifold $X\simeq S^1\times M$ with metric $\gamma:=N^2 dt^2\oplus g$, where $g$ and $N$ obey (\ref{eq1.1}, \ref{eq1.2}). This may be considered to be the ``Wick rotated'' spacetime built from $g$ and $N$. It is therefore Einstein, and
\begin{equation}
\label{eq4.1}
{\rm Ric}(\gamma)+n\gamma=0\ .
\end{equation}
As well, $\frac{\partial}{\partial t}$ is a global Killing vector field tangent to copies of the $S^1$ factor.

Denoting the coordinates by $(t,x^i)$, $i\in {1,\dots n}$, it is easy to compute that
\begin{equation}
\label{eq4.2}
\begin{split}
R_{0i0j}(\gamma)=&\, -N\nabla_i\nabla_j N\ ,\\
R_{0ijk}(\gamma)=&\, 0\ ,\\
R_{ijkl}(\gamma) =&\, R_{ijkl}(g)\ ,
\end{split}
\end{equation}
where $\nabla$ is the Levi-Civita connection of $g$. It follows that
\begin{equation}
\label{eq4.3}
\left \vert {\rm Riem}(\gamma) \right \vert_{\gamma}^2 = \frac{4}{N^2} \left \vert {\rm Hess} N \right \vert_g^2 +\left \vert {\rm Riem}(g) \right \vert_g^2\ .
\end{equation}
Using (\ref{eq1.1}), we have that ${\rm Hess}N\equiv \nabla^2 N =N\left ({\rm Ric}(g) + ng\right ) = \left ( Z(g)+g\right )N$, and so
\begin{equation}
\label{eq4.4}
\left \vert {\rm Riem}(\gamma) \right \vert_{\gamma}^2 = \left \vert {\rm Riem}(g) \right \vert_g^2+4\left \vert Z(g) \right \vert_g^2+4n\ .
\end{equation}
If we use the standard decomposition of ${\rm Riem}(g)$ into its Weyl, tracefree Ricci, and scalar curvature parts, and use that the scalar curvature of $g$ is $-n(n-1)$, we arrive at
\begin{equation}
\label{eq4.5}
\left \vert {\rm Riem}(\gamma) \right \vert_{\gamma}^2 = \left \vert {\rm Weyl}(g) \right \vert_g^2 +8\left \vert Z(g) \right \vert_g^2+2n(n+1)\ .
\end{equation}

Now choose $n=3$, so that $(X,\gamma)$ is a 4-manifold. Then the Pfaffian of the curvature 2-form of $(X,\gamma)$ is
\begin{equation}
\label{eq4.6}
{\rm Pfaff}(X,\gamma) = \left \vert {\rm Riem}(\gamma) \right \vert_{\gamma}^2 -4 \left \vert Z(\gamma) \right \vert_{\gamma}^2\ .
\end{equation}
We define the renormalized Pfaffian by
\begin{equation}
\label{eq4.7}
{\rm RenPf}(X,\gamma):={\rm Pfaff}(X,\gamma)-24\ .
\end{equation}
Using that $(X,\gamma)$ is Einstein, then $Z(\gamma)=0$. Furthermore, using (\ref{eq4.5}) and the fact that ${\rm Weyl}(g)=0$ since $(M,g)$ is a 3-manifold, then we get that
\begin{equation}
\label{eq4.8}
{\rm RenPf}(X,\gamma) = 8\left \vert Z(g) \right \vert_g^2\ .
\end{equation}

\subsection{The proof of Theorem \ref{theorem1.2} Part (a).}

\noindent
\begin{proof} Now Anderson's formula \cite{Anderson} for the \emph{renormalized volume} of a Poincar\'e-Einstein 4-manifold is
\begin{equation}
\label{eq4.9}
{\rm RenV}(X,\gamma)= \frac{4\pi^2}{3}\chi(X) -\frac{1}{24}\int_X {\rm RenPf}(X,\gamma) dV(\gamma)\ ,
\end{equation}
where\footnote
{We use tensor norms throughout, whereas the norms used in \cite{Anderson} are those used for forms. As a result, numerical coefficients in \cite[equations (0.1) and (1.26)]{Anderson} differ from those in equations (\ref{eq4.5}), (\ref{eq4.6}), and (\ref{eq4.10}).}
$\chi(X)$ is the Euler characteristic of $X$. Since $X$ is a product with an $S^1$ factor, $\chi(X)$ is zero here. Furthermore, ${\rm RenV}(X,\gamma)$ is independent of the coordinate $t$ for the $S^1$ factor, so the integral over $X$ becomes $\beta$ times an integral over $M$ with measure $NdV(g)$. Then from (\ref{eq4.8}) and (\ref{eq4.9}), we obtain
\begin{equation}
\label{eq4.10}
{\rm RenV}(X,\gamma)= -\frac{\beta}{3}\int_M \left \vert Z(g) \right \vert_g^2 NdV(g)
= \frac{8\pi}{3}\beta m \le 0 \ ,
\end{equation}
where the last equality follows from Theorem \ref{theorem1.1}. \end{proof}

\subsection{An example}

\noindent A rather nice example of Lemma \ref{theorem1.2} is afforded by the AdS soliton metric \cite{HM}
\begin{equation}
\label{eq4.11}
\gamma = r^2 dt^2+ \frac{dr^2}{r^2 \left ( 1-\frac{a^3}{r^3} \right )}+r^2 \left ( 1-\frac{a^3}{r^3} \right )d\xi^2+r^2d\theta^2\ ,
\end{equation}
with $a>0$, $r\in [a,\infty)$, and $\xi\in  [ 0, 4\pi/3 ]$. The domains of $t$ and $\theta$ can be chosen arbitrarily, so we take $\theta\in [0,\Theta]$ and, to conform with Theorem \ref{theorem1.2}, $t\in [0,\beta]$ (we need here that $x\sim 1/r$, as verified immediately below). The Killing vector of interest is $\frac{\partial}{\partial t}$. It is an easy matter to compute the mass of the $t=0$ slice using Wang's formula. One finds
\begin{equation}
\label{eq4.12}
m=-4\pi \Theta\beta/3\ .
\end{equation}
Note that it is negative.

To verify Theorem \ref{theorem1.2}, one can compute the renormalized volume of this manifold using Anderson's formula, but a more independent method is to compute it using Hadamard's regularization of the volume integral. To begin, we integrate $dx/x= dr/\left [ r\left ( 1-1/r^3\right ) \right ]$ to find a special defining function
\begin{equation}
\label{eq4.13}
x=\left ( r^{3/2}-\sqrt{r^3-1} \right )^{2/3}\ .
\end{equation}
Note that $x\in [0,1]$ and also that $x\sim 1/r$ as needed above. Using $x$ as a coordinate, we can write the metric as
\begin{equation}
\label{eq4.14}
\gamma= \frac{dx^2}{x^2}+\frac{1}{4^{2/3}}\frac{\left ( x^{-3/2}-x^{3/2} \right )^2}{\left ( x^{-3/2}+x^{3/2}\right )^{2/3}}d\xi^2 +\frac{1}{4^{2/3}}\left ( x^{-3/2}+x^{3/2}\right )^{4/3}\left ( dt^2+d\theta^2\right ) \ .
\end{equation}
More to the point, the volume element is
\begin{equation}
\label{eq4.15}
dV(\gamma)=\frac{1}{4x}\left ( x^{-3}-x^3\right ) dx dt d\xi d\theta\ ,
\end{equation}
keeping in mind that $x^{-3}\ge x^3$ for $x\in (0,1]$. We integrate this volume form using the above coordinate domains but truncating the $x$ domain to $x\in [\epsilon,1]$ for some $\epsilon>0$. Then
\begin{equation}
\label{eq4.16}
V_{\epsilon}= \frac{4\pi}{3}\Theta\beta \int_{\epsilon}^1 \frac{1}{4}\left ( x^{-4}-x^2 \right ) dx
= \frac{\pi}{9}\Theta\beta\left ( \frac{1}{\epsilon^3} - 2+\epsilon^3\right )\ .
\end{equation}
We take the Hadamard finite part as $\epsilon\to 0$, which means here that we simply remove the $1/\epsilon^3$ term before taking the limit. We get
\begin{equation}
\label{eq4.17}
{\rm RenV}(X,\gamma) = {\rm PF}_{\epsilon\to 0} V_{\epsilon}= -\frac{2\pi}{9}\Theta\beta\ ,
\end{equation}
so by comparing (\ref{eq4.17}) with (\ref{eq4.12}) we see that ${\rm RenV}(X,\gamma)=m\beta/6$, as required.

\subsection{The proof of Theorem \ref{theorem1.2} Part (b).}

\noindent
\begin{proof} Assume now that $N$ has a non-empty zero set which is a closed, connected hypersurface ${\cal H}$ with surface gravity $\vartheta=|dN|_{\cal H}$. Equations (\ref{eq4.8}), (\ref{eq4.9}), and (\ref{eq1.8}) immediately yield
\begin{equation}
\label{eq4.18}
\begin{split}
{\rm RenV}(X,\gamma)= &\, \frac{4\pi^2}{3} \left [ \chi(X)-\frac{\beta\vartheta}{2\pi}\chi({\cal H})\right ] +\frac{8\pi}{3} \left [ m\beta-\frac{\beta\vartheta}{2\pi}\frac{\vert{\cal H}\vert}{4}\right ]\\
= &\, \frac{4\pi^2}{3} \left [ 1-\frac{\beta\vartheta}{2\pi}\right ]\chi({\cal H})+\frac{8\pi}{3} \left [ m\beta-\frac{\beta\vartheta}{2\pi}\frac{\vert{\cal H}\vert}{4}\right ]\ ,
\end{split}
\end{equation}
where $\beta$ is the circumference of the Killing orbits at infinity (in the conformal metric). and where we used that $X\simeq {\mathbb R}^2\times {\cal H}$ so that $\chi(X)=\chi({\cal H})$.

In a neighbourhood of ${\cal H}$, we can use Gaussian normal coordinates $(\rho,x^a)$ for the metric $g$. Then the metric $\gamma=N^2 dt^2\oplus g$ can be written as
\begin{equation}
\label{eq4.19}
\gamma=N^2(\rho,x^c) dt^2 +d\rho^2+h_{ab}(\rho,x^c) dx^adx^b\ ,
\end{equation}
where $N(0,x^c)=0$ and $h_{ab}(0,x^c)$ is the metric induced on ${\cal H}$. Recall that we identify $t\sim t+\beta$ and we have $\vert dN\vert_{\cal H} =: \vartheta \neq 0$. Then $\vartheta=\frac{\partial N}{\partial \rho}(0,x^c)\neq 0$. Smoothness of $\gamma$ then requires that $d\rho^2+N^2 dt^2\sim d\rho^2+\rho^2 d\xi^2$ at $\rho=0$, with $\xi\in [0,2\pi]$. Therefore, $\vartheta=\frac{2\pi}{\beta}$. Inserting this in (\ref{eq4.18}), the coefficient of $\chi({\cal H})$ vanishes and the coefficient of $|{\cal H}|/4$ simplifies to $1$, yielding equation (\ref{eq1.10}). \end{proof}

\section{Open problems}
\setcounter{equation}{0}

\noindent We close by discussing some open problems. One that was raised a long time ago and which remains open to this day is prompted by the example of Section 4.3 except with $dt^2$ replaced by $-dt^2$ in equation (4.10). These are the AdS soliton metrics \cite{HM}, which are nontrivial, globally static, negative mass solutions of (\ref{eq1.1}, \ref{eq1.2}) with toroidal ($k=0$) boundary at infinity. We will refer to the static slices as \emph{Horowitz-Myers geons} \cite{BW}. These slices have negative mass. They evade the positive mass theorem for spin manifolds, even though they are spin, because the spinor structure does not admit asymptotically constant solutions of the Witten equation. It is a conjecture of Horowitz and Myers that, amongst all APEs with scalar curvature $R\ge -n(n-1)$ and the same flat torus as conformal infinity, the minimizer of the Wang mass is a Horowitz-Myers geon.\footnote
{Note that in dimension $n\ge 4$ there is more than one such Horowitz-Myers geon filling in a given toroidal boundary-at-infinity. The conjectural minimizer is constructed by attaching the bulk manifold to the boundary torus in such a manner that the shortest cycle on the torus becomes contractible to a point in the interior.}
The original conjecture was motivated by the AdS/CFT correspondence and gauge theory arguments, and because of this it was posed in the fixed dimension $n=4$, but it seems to us no less plausible in other dimensions.

\medskip
\noindent {\bf Problem 5.1.} Choose a fixed flat $n-1$ torus, $n\ge 3$. Find a complete APE $n$-manifold with scalar curvature $R\ge = -n(n-1)$, with inner boundary empty or a compact minimal hypersurface, which has this flat torus as its boundary-at-infinity and which has Wang mass less than that of any Horowitz-Myers geon with the same boundary-at-infinity, or prove that no such $n$-manifold exists. This is basically a generalization in dimension of the problem of proving or disproving Conjecture 3 of \cite{HM}.
\medskip

\noindent In this regard, we recall that Gibbons once attempted to address this conjecture using the monotonicity of a generalized Hawking mass under the inverse mean curvature flow \cite{Gibbons}. He fixed $n=3$ and studied the behaviour of a quantity that can be written using our conventions as
\begin{equation}
\label{eq5.1}
m_H(\Sigma):=\frac{\vert \Sigma \vert^{1/2}}{64\pi^{3/2}} \int_{\Sigma}\left ( 2S-H^2+4\right )dV(h) \ ,
\end{equation}
where $\Sigma$ is a closed embedded hypersurface and $\vert \Sigma\vert :=\int_{\Sigma}dV(h)$. This quantity exhibits \emph{Geroch monotonicity} \cite{Geroch}: it is monotonic under mean curvature flow, so that $m_H(\Sigma_2)\ge m_H(\Sigma_1)$ whenever $\Sigma_2$ is obtained by evolving $\Sigma_1$ outward (toward infinity) by mean curvature flow. Gibbons's technique failed to resolve the conjecture because this quantity diverges to $-\infty$ on the central circle of the AdS soliton.

For sake of comparison, let $Q(\Sigma)$ denote $-\frac{1}{8\pi^{3/2}}$ times the quantity on the left-hand side of (\ref{eq3.7}), integrated over $\Sigma$. Again, we set $n=3$. Using (\ref{eq3.7}) and (\ref{eq3.9}), we obtain
\begin{equation}
\label{eq5.2} Q(\Sigma)= \frac{1}{64\pi^{3/2}}\int_{\Sigma} \left ( 2S-H^2+4+2|A^{\rm TF}|\right )|dN| dV(h)\ ,
\end{equation}
where $A^{\rm TF}$ is the trace-free part of the second fundament form $A$, and by integrating (\ref{eq3.2}) we have the monotonicity $Q(\Sigma_2)\ge Q(\Sigma_1)$ \emph{whenever} $\Sigma_2$ lies between $\Sigma_1$ and the boundary-at-infinity, without regard to whether these surfaces are related by mean curvature flow. Note that (\ref{eq5.1}) and (\ref{eq5.2}) are quite similar, but where $\ref{eq5.1}$ has $|\Sigma |^{1/2}$, (\ref{eq5.2}) has $|dN|$. This seems to help, because it implies that $Q(\Sigma)$ tends to zero as $\Sigma$ shrinks down to the central circle, in contrast to $m_H$. Also, (\ref{eq5.2}) contains an $A^{\rm TF}$ term, which seems to be a necessary modification to the Geroch monotonicity argument, because this term will not be zero for Horowitz-Myers geons. However, as $\Sigma$ approaches conformal infinity, $Q$ approaches the \emph{negative} of the mass (times a positive constant), yielding only Theorem \ref{theorem1.1} and not a proof of the Horowitz-Myers conjecture.

In light of the conjecture of Horowitz and Myers, then it seems reasonable to wonder whether there is a version when the boundary-at-infinity is a compact hyperbolic surface; i.e., when $k=-1$. It is known that there are static metrics with a horizon in this case \cite{Mann, BLP}, and they form families with mass bounded below, but the mass does become negative along these families. The lower bound is realized by a cold horizon. It is not known if there are any complete metrics of this form without a horizon.

\medskip
\noindent {\bf Problem 5.2.} Choose a fixed closed, orientable, connected, hyperbolic $(n-1)$-manifold $\Sigma$, $n\ge 3$. Find a complete, boundaryless $n$-manifold $(M,g)$ and a positive function $N$ such that
\begin{enumerate}
\item $(M,g)$ is APE, with boundary-at-infinity isometric to $\Sigma$, and
\item $(M,g,N)$ obeys the system $(\ref{eq1.1}, \ref{eq1.2})$ (equivalently, (\ref{eq1.4})), with $|dN|\to 0$ at infinity,
\end{enumerate}
or show that no such $(M,g,N)$ exists.
\medskip

\noindent If the boundary at infinity were a torus, the solution of this problem of course would be a Horowitz-Myers geon $(M,g)$, together with $N$ such that $-N^2 dt^2\oplus g$ is an AdS soliton.

We know that any metric solving Problem 5.2 must have negative Wang mass, or zero Wang mass if $N\equiv 0$ since the special case of $N\equiv 0$ would be a Poincar\'e-Einstein manifold. We do not know if such a Poincar\'e-Einstein manifold exists when $k=-1$.

\medskip
\noindent \noindent {\bf Problem 5.3.} Choose a fixed closed, orientable, connected, hyperbolic $(n-1)$-manifold $\Sigma$, $n\ge 3$. Find a complete, boundaryless Poincar\'e-Einstein $n$-manifold $(M,g)$ with boundary-at-infinity isometric to $\Sigma$, or show that no such $(M,g)$ exists.
\medskip

\noindent If $k$ were $1$, the solution of this problem would be standard hyperbolic $n$-space. The Riemannian AdS solitons (\ref{eq4.11}) provide nontrivial solutions when $k=0$. Thus the $k=-1$ case is intriguing.

An obvious direction in which to attempt to generalize the mass formula, the renormalized volume calculation, and some of the problems listed above is to attempt to pass to the stationary but non-static case. In this case, the system (\ref{eq1.1}, \ref{eq1.2}) is replaced by a somewhat more complicated system which we will not write down, whose solution includes a metric $g$ and lapse function and an additional vector field, the so-called \emph{shift vector}. The Killing vector field $\frac{\partial}{\partial t}$ is no longer hypersurface orthogonal, but the quotient of spacetime by the vector field is smooth and it is on this quotient that $g$ is a metric. An obvious first question is

\medskip
\noindent \noindent {\bf Problem 5.4.} What are the generalizations of Theorems \ref{theorem1.1} and \ref{theorem1.2} that apply to solutions stationary Einstein equations?
\medskip

\noindent One would expect such generalizations to look like the formulas of Theorems \ref{theorem1.1} and \ref{theorem1.2}, but with additional terms which depend on the curl of the shift vector.

In the spirit of Problem 5.2, and motivated by a classic result of Lichn\'erowicz \cite[p 142]{Lichnerowicz} for asymptototically flat spacetimes, an intriguing question concerning the AdS solitons is whether it is possible to endow them with rotation about the central axis (the locus $r=1$ in (\ref{eq4.11}) with $dt^2$ replaced by $-dt^2$), without formation of either a horizon or a naked singularity. The resulting spacetime would be a nontrivial stationary AdS soliton.

\medskip
\noindent {\bf Problem 5.5.} Are there families of $(n+1)$-dimensional, stationary vacuum spacetimes which include an AdS soliton and are asymptotically locally anti-de Sitter with the same conformal infinity as this soliton?

\end{document}